\documentclass{article}

\usepackage{microtype, subfigure, verbatim}
\usepackage{booktabs, hyperref}
\usepackage{amssymb, amsthm, graphicx, amsmath,multirow}
\usepackage{algorithm, algorithmic, wrapfig}
\usepackage[inline]{enumitem}
\usepackage[small]{caption}

\usepackage[accepted]{icml2020}

\def\x{{\mathbf x}}  
\def\y{{\mathbf y}} \def\z{{\mathbf z}} \def\by{{\bar{\mathbf y}}} \def\bT{{\bar{T}}} 
  \def\tt{{\tilde t}}
\def\L{{\cal L}} \def\D{{\cal D}}
\def\p{{\mathbf p}} \def\hp{{\hat{\mathbf p}}} \def\l{{\ell}}
\def\N{{\mathcal N}}

\newtheorem{proposition}{Proposition}
\newtheorem{corollary}{Corollary}

\renewcommand{\algorithmicrequire}{\textbf{Input:}}
\renewcommand{\algorithmicensure}{\textbf{Output:}}

\icmltitlerunning{Sequential Self Teaching for Learning Sounds}

\begin{document}

\twocolumn[
\icmltitle{A Sequential Self Teaching Approach for \\Improving Generalization in Sound Event Recognition}

\begin{icmlauthorlist}
\icmlauthor{Anurag Kumar}{frl}
\icmlauthor{Vamsi  Krishna Ithapu}{frl}
\end{icmlauthorlist}

\icmlaffiliation{frl}{Facebook Reality Labs, Redmond, USA}
\icmlcorrespondingauthor{Anurag Kumar}{anuragkr@fb.com}

\icmlkeywords{Sequential Learning, Noisy Labels, Audio Event Detection}

\vskip 0.3in
]
\printAffiliationsAndNotice{}

%%%%%%%%%%%%%%%%%%%%%%%%%%%%%%%%%%%%%%%%%%%%%%%%%%%%%%%%%%%%%%%%%%%%%%%%%%%%%%%%%%%%%%%%

\begin{abstract}
An important problem in machine auditory perception is to recognize and detect sound events. In this paper, we propose a sequential self-teaching approach to learning sounds. Our main proposition is that it is harder to learn sounds in adverse situations such as from weakly labeled and/or noisy labeled data,  and in these situations a single stage of learning is not sufficient. Our proposal is a sequential stage-wise learning process that improves generalization capabilities of a given modeling system. We justify this method via technical results and on Audioset, the largest sound events dataset, our sequential learning approach can lead to up to 9\% improvement in performance. A comprehensive evaluation also shows that the method leads to improved transferability of knowledge from previously trained models, thereby leading to improved generalization capabilities on transfer learning tasks. 
\end{abstract}

%%%%%%%%%%%%%%%%%%%%%%%%%%%%%%%%% Introduction %%%%%%%%%%%%%%%%%%%%%%%%%%%%%%%%%%%%%%%%%

\section{Introduction} 
\label{sec:intro}

Human interaction with the environment is driven by multi-sensory perception. 
Sounds and sound events, natural or otherwise, play a vital role in this first person interaction. 
To that end, it is  imperative that we build acoustically intelligent devices and systems which can {\it recognize } and {\it understand} sounds. 
Although this aspect has been identified to an extent, 
and the field of Sound Event Recognition and detection (SER) is at least a couple of decades old~\cite{xiong2003audio,atrey2006audio}, 
much of the progress has been in the last few years~\cite{virtanen2018computational}. 
Similar to related research domains in machine perception, like speech recognition, 
most of the early works in SER were fully supervised and driven by \emph{strongly labeled data}. 
Here, audio recordings were carefully (and meticulously) annotated with time stamps of sound events to produce exemplars.
These exemplars then drive the training modules in supervised learning methods. 
Clearly, obtaining well annotated strongly labeled data is prohibitively expensive and cannot be scaled in practice. 
Hence, much of the recent progress on SER has focused on efficiently leveraging \emph{weakly labeled} data \cite{kumar2016audio}. 

Weakly labeled audio recordings are only tagged with presence or absence of sounds (i.e.,  a binary label), 
and no temporal information about the event is provided. 
Although this has played a crucial role in scaling SER, 
large scale learning of sounds remains a challenging and open problem. 
This is mainly because, even in the presence of strong labels, large scale SER brings adverse learning conditions into the picture, 
either implicitly by design or explicitly because of the sheer number and variety of classes. 
This becomes more critical when we replace strong labels with weak labels.  
Tagging ({\it a.k.a.} weak labeling) very large number of sound categories in large number of recordings, 
often leads to considerable label noise in the training data. This is expected. 
Implicit noise via human annotation errors is clearly one of the primary factors contributing to this.
\emph{Audioset}~\cite{gemmeke2017audio}, currently the largest sound event dataset, suffers from this implicit label noise issue.
Correcting for this implicit noise is naturally very expensive (one has to perform multiple re-labeling of the same dataset). 
Beyond this, there are more nuanced noise inducing attributes, which are outcomes of the number and variance of the classes themselves. 
For instance, real world sound events often overlap and as we increase the sound vocabulary and audio data size, 
the ``mass'' of overlapping audio in the training data can become large enough to start affecting the learning process. 
This is trickier to address in weakly labeled data where temporal locations of events are not available. 

Lastly, when working with large real world datasets, one cannot avoid the noise in the inputs themselves.
For SER these manifest either via signal corruption in the audio snippets themselves (i.e., acoustic noise), 
or signals from non-target sound events, both of which will interfere in the learning process. 
In weakly labeled setting, by definition, this noise level would be high, presenting harsher learning space for networks. 
We need efficient SER methods that are sufficiently robust to the above three adverse learning conditions. 
In this work, we present an interesting take on large scale weakly supervised learning for sound events.
Although we focus on SER in this work, we expect that 
the proposed framework is applicable for any supervised learning task.

The main idea behind our proposed framework is motivated by the attributes of human learning, 
and how humans adapt and learn when solving new tasks. 
An important characteristic of human's ability to learn is that it is {\it not} a one-shot learning process, 
i.e., in general, we do not learn to solve a task in the first attempt. 
Our learning typically involves multiple stages of development where past experiences, and past failures or successes, 
``guide'' the learning process at any given time. 
This idea of sequential learning in humans wherein each stage of learning is guided by previous stage(s) was referred to as \emph{sequence of teaching selves} in \cite{minsky1994society}. 
Our proposal follows this meta principle of sequential learning, and at the core, it involves the concept of learning over time. 
Observe that this learning over time is rather different from, for instance, 
learning over iterations or epochs in stochastic gradients, and we make this distinction clear as we present our model. 
We also note that the notion of lifelong learning in humans, which has inspired \emph{lifelong machine learning}~\cite{silver2013lifelong,parisi2019continual}, is also, in principle, related to our framework. 

Our proposed framework is called SeqUential Self TeAchINg (SUSTAIN). 
We train a sequence of neural networks (designed for weakly labeled audio data) wherein 
the network at the current stage is guided by trained network(s) from the previous stage(s). 
The guidance from networks in previous stages comes in the form of ``co-supervision''; 
i.e., the current stage network is trained using a \emph{convex combination} of ground truth labels and the outputs from one or more networks from the previous stages. 
Clearly, this leads to a cascade of {\it teacher-student} networks. 
The student network trained in the current stage will become a teacher in the future stages. 
We note that this is also related to the recent work on knowledge distillation through teacher-student frameworks \cite{hinton2015distilling, ba2014deep, bucilua2006model}. 
However, unlike these, our aim is not to construct a smaller, compressed, model that emulates the performance of high capacity teacher. 
Instead, our SUSTAIN framework's goal is to simply utilize the teacher`s knowledge better. 

Specifically, the student network tries to correct the mistakes of the teachers, 
and this happens over multiple sequential stages of training and co-supervision with the aim of building better models as time progresses. 
We show that one can quantify the performance improvement, 
by explicitly controlling the transfer of knowledge from teacher to student over successive stages.

The {\bf contributions} of this work include: 
\begin{enumerate*}[label=\textbf{(\alph*)}]
\item A sequential self-teaching framework based on co-supervision for improving learning over time, 
including few technical results characterizing the limits of this improved learnability;  
\item A novel CNN for large scale weakly labeled SER, and
\item Extensive evaluations of the framework showing up to $9\%$ performance 
improvement on Audioset,  significantly outperforming existing procedures, and applicability to knowledge transfer. 
\end{enumerate*}

The rest of the paper is organized as follows. We discuss some related work in Section \ref{sec:rwork}.  In Section \ref{sec:sustain}, we introduce the sequential self-teaching framework and then discuss few technical results. 
In Section \ref{sec:weakcnn}, we describe our novel CNN architecture for SER which learns from weakly labeled audio data. 
Sections \ref{sec:expt} and \ref{sec:transfer} show our experimental results, and we conclude in Section \ref{sec:conc}. 

%%%%%%%%%%%%%%%%%%%%%%%%%%%%%%% Related Work %%%%%%%%%%%%%%%%%%%%%%%%%%%%%%%%%%%%%%%%

\section{Related Work} \label{sec:rwork}

While earlier works on SER were primarily small scale~\cite{couvreur1998automatic}, large scale SER has received considerable attention in the last few years. 
The possibility of learning from weakly labeled data~\cite{kumar2016audio, su2017weakly} is the primary driver here, 
including availability of large scale weakly labeled datasets later on, like Audioset ~\cite{gemmeke2017audio}. 
Several methods have been proposed for weakly labeled SER; \cite{kumar2017knowledge,kong2019weakly,chou2018learning,mcfee2018adaptive,yu2018multi,wang2018comparing,adavanne2017sound} to name a few. 
Most of these works employ deep convolutional neural networks (CNN). 
The inputs to CNNs are often time-frequency representations such as spectrograms, logmel spectrograms, constant-q spectrograms~\cite{zhang2015robust,kumar2017knowledge,ye2015acoustic}. 
Specifically, with respect to Audioset, some prior works, for example ~\cite{kong2019weakly}, have used features from a pre-trained network, trained on a massive amount of YouTube data~\cite{hershey2017cnn} for instance.

The weak label component of the learning process was earlier handled via $mean$ or $max$ global pooling~\cite{su2017weakly,kumar2017knowledge}. 
Recently, several authors proposed to use attention~\cite{kong2019weakly,wang2018comparing,chen2018class}, recurrent neural networks~\cite{adavanne2017sound}, adaptive pooling~\cite{mcfee2018adaptive}. 
Some works have tried to understand adverse learning conditions in weakly supervised learning of sounds~\cite{shah2018closer, kumar2019learning}, although it still is an open problem. 
Recently, problems related to learning from noisy labels have been included in the annual DCASE challenge on sound event classification~\cite{fonseca2019audio} \footnote{http://dcase.community/challenge2019/}. 

Sequential learning, and more generally, learning over time, 
is being actively studied recently~\cite{parisi2019continual}, 
starting from the seminal work~\cite{minsky1994society}. 
Building cascades of models has also been tied to lifelong learning~\cite{silver2013lifelong,ruvolo2013ella}.
Further, several authors have looked at the teacher-student paradigm in a variety of contexts including 
knowledge distillation~\cite{hinton2015distilling,furlanello2018born,chen2017learning,mirzadeh2019improved}, compression~\cite{polino2018model} and transfer learning~\cite{Yim_2017_CVPR,weinshall2018curriculum}. \cite{furlanello2018born} in particular show that it is possible to sequentially distill knowledge from neural networks and improve performance. Our work builds on top of \cite{kumar2020secost}, and proposes to learn a sequence of self-teacher(s) to improve generalizability  in adverse learning conditions. This is done by co-supervising the network training along with available labels and controlling the knowledge transfer from the teacher to the student. %We show improvements in generalizability (both empirically and through technical results) of learned models in these adverse learning conditions. 

%%%%%%%%%%%%%%%%%%%%%%%%%%%%%%%%% Some Theory %%%%%%%%%%%%%%%%%%%%%%%%%%%%%%%%%%%%%%%%%

\section{Sequential Self-Teaching (SUSTAIN)} \label{sec:sustain}

\subsection{SUSTAIN Framework} \label{sec:framework}

\paragraph{Notation} \label{para:notation}: 
Let $\D := \{ \x^s, \y^s\}$ ($s=1,\ldots,S$) denote the dataset we want to learn with $S$ training pairs. $\x^s$ are the inputs to the learning algorithms and  $\y^s \in \{0,1\}^C$ are the desired outputs. $C$ is the number of classes.
$y^s_c = 1$ indicates the presence of $c^{th}$ class in the input $\x^s$. Note that $\y^s_c \;\forall\; c$ are the observed labels and may have noise. 

For the rest of the paper, we restrict ourselves to the binary cross-entropy loss function. However, in general, the method is applicable to other loss functions as well, such as mean squared error loss. 

If $\p^s = [p^s_1,\ldots,p^s_C]$ is the predicted output, then the loss is
\begin{align} \label{eq:crossent}
\L(\p^s, \y^s) &= \frac{1}{C}\sum_{c=1}^{C} \l(p^s_c, y^s_c) \hspace{2mm} \text{where} \\ 
\l(p^s_c, y^s_c) &= - y^s_c \log(p^s_c) - (1-y^s_c) \log(1-p^s_c)
\end{align}

With this notation, we will now formalize the ideas motivated in Section \ref{sec:intro}.
The learning process entails $T$ stages indexed by $t = 0,\ldots,T$. 
The goal is to train a cascade of learning models denoted by $\N^0,\ldots,\N^T$ at each stage. 
The final model of interest is $\N^T$. 
Zeroth stage serves as an {\it initialization} for this cascade. 
It is the default teacher that learns from the available labels $\y^s$. Once $\N^0$ is trained, we can get the predictions $\hp^s_0 \; \forall \; s$ (note the $\hat{\cdot}$ here). 

The learning in each of the later stages is co-supervised by the already trained network(s) from previous stages, 
i.e., at $t^{th}$ stage, $\mathcal{N}^0, \dots\mathcal{N}^{t-1}$ guide the training of $\N^t$. 
This guidance is done via replacing the original labels ($\y^s$) with a convex combination of the predictions from the teacher network(s) and $\y^s$, which will be the new targets for training $\N^t$. 
In the most general case, if all networks from previous stages are used for teaching, the new target at $t^{th}$ stage is, 
\begin{equation} \label{eq:new-labels}
\by_t^s = \alpha_0 \y^s + \sum_{\tt=1}^t \alpha_\tt \hp_{\tt-1}^s \quad \text{\it s.t.} \quad \sum_{\tt=0}^{t} \alpha_\tt = 1 
\end{equation}

More practically, the network from only last stage will be used, in which case, 
\begin{equation} \label{eq:new-labels-1-teacher}
\by_t^s = \alpha_0 \y^s + (1 - \alpha_0) \hp_{t-1}^s    
\end{equation}
or the students from previous $m$ stages will co-supervise the learning at stage $t$, which will lead to $\by_t^s = \alpha_0 \y^s + \sum_{\tt=1}^m \alpha_\tt \hp_{t - \tt}^s $, $\,s.t\,$ $ \sum_{\tt=0}^{m} \alpha_\tt = 1 $. 

Algorithm \ref{alg:secost} summarizes this self teaching approach driven by co-supervision with single teacher per stage. It is easy to extend it to $m$ teachers per stage, driven by appropriately chosen $\alpha$`s. 

\setlength{\textfloatsep}{5pt}
\begin{algorithm}[t!]
	\caption{\textbf{SUSTAIN}: Single Teacher Per Stage}
	\begin{algorithmic}[1]
		\REQUIRE: $\D$, $\#$stages $T$, \{$\alpha_t$, $t = 0,\ldots,T-1$ \} \\
		\ENSURE: Trained Network $\N^T$ after $T$ stages
		\STATE Train default teacher $\N^0$ using $\D := \{ \x^s, \y^s\}$ $\forall \; s$ 
		\FOR{$t = 1,\ldots,T$}
		\STATE Compute new target $\by_t^s$ ($\forall s$) using Eq. \ref{eq:new-labels-1-teacher}% (single teacher)
		\STATE Train $\N^t$ using new target $\D := \{ \x^s, \by_t^s\}$ $\forall \; s$ 
		\ENDFOR
		\STATE Return $\N^T$
	\end{algorithmic}
	\label{alg:secost}
\end{algorithm}

%%%%%%%%%%%%%%%%%%%%%%%%%

\subsection{Analyzing SUSTAIN w.r.t to label noise} \label{sec:analysis}

In this section, we provide some insights into our SUSTAIN method with respect to label noise, a common problem in large scale learning of sound events. $\y^s_c \; \forall \; c$ denote our noisy observed labels. Let $y^{*s}_c$ be the corresponding true label parameterized as follows,
\begin{equation} \label{eq:true-labels}
y^s_c = \begin{cases} 
y^{*s}_c & \text{\it w.p.} \hspace{2mm} \delta_c \\ 
1-y^{*s}_c & \text{\it else} \end{cases}
\end{equation}  
Within the context of learning sounds, in the simplest case, $\delta_c$ characterizes the per-class noise in labeling process. Nevertheless, depending on the nature of the labels themselves, it may represent something more general like sensor noise, overlapping speakers and sounds etc. 

To analyze our approach and to derive some technical guarantees on performance, we assume a trained default teacher $\N^0$ and a new student to be learned (i.e., $T=1$). The new training targets in this case are given by
\begin{equation} \label{eq:new-label-1-stage}
\by_1^s = \alpha_0 \y^s + (1-\alpha_0) \hp_0^s \\
\end{equation}

Recall from Eq. \ref{eq:true-labels} that $\delta_c$ parameterizes the error in $\y^s$ vs. the unknown truth $\y^{*s}$.
Similarly, we define $\bar{\delta}_c$ to parameterize the error in $\hp^s$ vs. $\y^{*s}$ i.e., noise in teacher's predictions w.r.t the true unobserved labels. 

\begin{equation} \label{eq:one-teacher-delta-bar-def}
\hat{p}^s_{0,c} = \begin{cases} 
y^{*s}_c & \text{\it w.p.} \hspace{2mm} \bar{\delta_c} \\ 
1-y^{*s}_c & \text{\it else} \end{cases}
\end{equation}  
The interplay between $\delta_c$ and $\bar{\delta}_c$ in tandem with the performance accuracy of $\N^0$ will help us evaluate the gain in performance for $\N^1$ versus $\N^0$. 
To theoretically assess this performance gain, we consider the case of uniform noise $\forall \; c$ followed by a commentary on class-dependent noise. 
Further, we explicitly focus the technical results on high noise setting and revisit the low-to-medium noise setup in evaluations in Section \ref{sec:expt}. 

\subsubsection{Uniform Noise: $\delta_c = \delta \; \forall \; c$} \label{sec:same-noise}

This is the simpler setting where the apriori noise in classes is uniform across all categories with $\delta_c = \delta \;\forall \;c$. We have the following result.

\begin{proposition} \label{prop:delta-bar}
  Let $\N^1$ be trained using $\{\x^s, \by^s\} \;\forall \;s$ using binary cross-entropy loss, and let $\epsilon_c$ denote the average accuracy of $\N^0$ for class $c$.
  Then, we have
  \begin{equation} \label{eq:one-teacher-delta-bar-val}
    \bar{\delta}_c = \epsilon_c\delta + (1-\epsilon_c)(1-\delta) \; \forall \; c
  \end{equation}  
  and whenever $\delta < \frac{1}{2}$, $\N^1$ improves performance over $\N^0$. The per class performance gain is $(1-\epsilon_c)(1-2\delta)$ 
\end{proposition}
\begin{proof}
Recall the entropy loss from Eq. \ref{eq:crossent}, for a given $s$ and $c$. 
Using the definition of the new label from Eq. \ref{eq:new-label-1-stage}, we get the following 
\begin{equation} \label{eq:prop-1}
    \l(p^s_c, \bar{y}^s_c) = \alpha_0 \l(p^s_c, y^s_c) + (1-\alpha_0) \l(p^s_c, \hat{p}^s_c)  
\end{equation}
Now, Eq. \ref{eq:true-labels} says that {\it w.p.} $\delta$ (recall $\delta_c = \delta \; \forall \; c$ here), $\l(p^s_c, y^s_c) = \l(p^s_c, y^{*s}_c)$, else $\l(p^s_c, y^s_c) = \l(p^s_c, 1-y^{*s}_c)$. 
Hence, using Eq. \ref{eq:true-labels} and Eq. \ref{eq:one-teacher-delta-bar-def}, and using the resulting equations in Eq. \ref{eq:prop-1} we have the following
\begin{align*}
    &\mathop{\mathbb{E}}_s \l(p^s_c, y^s_c) = \delta \sum_{s=1}^S \l(p^s_c, y^{*s}_c) + (1-\delta) \sum_{s=1}^S \l(p^s_c, 1-y^{*s}_c) \\
    &\mathop{\mathbb{E}}_s \l(p^s_c, \hat{p}^s_c) = \bar{\delta}_c \sum_{s=1}^S \l(p^s_c, y^{*s}_c) + (1-\bar{\delta}_c) \sum_{s=1}^S \l(p^s_c, 1-y^{*s}_c) \\
    &\mathop{\mathbb{E}}_s \l(p^s_c, \bar{y}^s_c) = (\alpha_0\delta + (1-\alpha_0)\bar{\delta}_c) \sum_{s=1}^S \l(p^s_c, y^{*s}_c) \\ &\quad + (\alpha_0(1-\delta) + (1-\alpha_0)(1-\bar{\delta}_c)) \sum_{s=1}^S \l(p^s_c, 1-y^{*s}_c)
\end{align*}
If $(\alpha_0\delta + (1-\alpha_0)\bar{\delta}_c) > \delta$ then we can ensure that using $\bar{y}^s_c$ as targets is better than using $y^s_c$. 
Now given the accuracy of $\N^0$ denoted by $\epsilon_c \; \forall \; c$, combining Eq. \ref{eq:true-labels} and Eq. \ref{eq:one-teacher-delta-bar-def}, we can see that $\bar{\delta}_c = \epsilon_c\delta + (1-\epsilon_c)(1-\delta)$. Using this, for $\N^1$ to be better than $\N^0$,  we need
\begin{equation} \label{eq:prop-2}
    \alpha_0\delta + (1-\alpha_0)(\epsilon_c\delta + (1-\epsilon_c)(1-\delta)) > \delta
\end{equation}
which requires $\delta < \frac{1}{2}$. And the gain is simply $\alpha_0\delta + (1-\alpha_0)\bar{\delta}_c - \delta$ which reduces to $(1-\epsilon_c)(1-2\delta)$. 
\end{proof}

\subsubsection{Remarks} \label{sec:remarks}

The above proposition is fairly intuitive and summarizes a core aspect of the proposed framework. 
Observe that, Proposition \ref{prop:delta-bar} is rather conservative in the sense that we are claiming $\N^1$ is better than $\N^0$ only if Eq. \ref{eq:prop-2} holds for all classes, 
i.e., performance improves for all classes. 
This may be relaxed, and we may care more about some specific classes. 
We discuss this below, for the high and low noise scenarios separately. 

\noindent \paragraph{High noise $\boldsymbol{\delta} < \frac{\mathbf{1}}{\mathbf{2}}:$} \label{sec:high-noise}
The given labels $y^s_c$ are wrong more than half of the time, 
and with such high noise, we expect $\N^0$ to have high error i.e., $\hat{p}^s_{0,c}$ and $y^s_c$ do not match.
Putting these together, as Proposition \ref{prop:delta-bar} suggests, 
the probability that $\hat{p}^s_{0,c}$ matches the truth $y^{*s}_c$ is implicitly large, 
leading to $\bar{\delta}_c > \delta$. 
Note that we cannot just flip {\bf all} predictions i.e., $p^s_{0,c} = 1-y^s_c$ would be infeasible, 
and there is some trade-off between $\N^0$'s predictions and given labels. 
Thereby, the choice of $\alpha_0$ then becomes critical (which we discuss further in Section \ref{sec:alpha-vs-T}).
Beyond this interpretation, we show extensive results in Section \ref{sec:expt} supporting this. 

\paragraph{Low-to-medium noise $\boldsymbol{\delta} > \frac{\mathbf{1}}{\mathbf{2}}:$}  \label{sec:low-noise}
When $\delta \gg \frac{1}{2}$, $\N^0$ is expected to perform well, 
and $\hat{p}^s_{0,c}$ matches $y^s_c$, which in turn matches $y^{*s}_c$ since the noise is low. 
Hence, $\N^1$'s role of combining $\N^0$'s output with $y^s_c$ becomes rather moot, because on average, for most cases, they are same. 
For medium noise settings with $1 \gg \delta > \frac{1}{2}$, proposition \ref{prop:delta-bar} does not infer anything specific. 
Nevertheless, via extensive set of experiments, we show in section \ref{sec:expt} that $\N^1$ still improves over $\N^0$ in some cases. 

\noindent \paragraph{Class-Specific Noise: $\delta_c \neq \delta \; \forall \; c$} \label{sec:different-noise}
It is reasonable to assume that in practice there are specific classes of interest that we desire to be 
more accurately predictable than others, including the fact that annotation is more carefully done for such classes. 
One can generalize Proposition \ref{prop:delta-bar} for this class-dependent $\delta_c$s, 
by putting some reasonable lower bound on loss of accuracy for undesired classes $\epsilon_c$s. 
We leave such technical details to a follow-up work, and now address the issue of choosing $\alpha$s for learning.  

\subsubsection{Interplay of $\alpha_t$ and $T$} \label{sec:alpha-vs-T}

Recall that the main hyperparameters of SUSTAIN are the weights $\alpha_0,\ldots,\alpha_T$ and $T$, 
and the main unknowns are the noise levels in the dataset ($\delta_c$).
%As described in Section \ref{sec:remarks}, $\delta_c$ governs the gain in performance.  
We now suggest that Algorithm \ref{alg:secost} is implicitly robust to these unknowns and provides an empirical strategy to choose the hyperparameters as well. 
We have the following result focusing on a given class $c$. $\bT_c$ and $\bT$ denote the optimal number of stages per class $c$ and across all classes respectively. The proof is in supplement. 

\begin{corollary} \label{cor:optimal-T}
Let $\epsilon_c^t$ denote the accuracy of $\N^t$ for class $c$. 
Given some $\delta$, there exists an optimal $\bT^c$
%and a corresponding sequence of weights $\alpha_0,\ldots,\alpha_{\bT}$, 
such that $\epsilon_c^{\bT^c} \geq \epsilon_c^t$. 
\end{corollary}

\noindent \paragraph{Remarks.} 
The main observation here is that $\bT_c$ might be very different for each $c$, 
and it may be possible that $\bT_c = 0$ in certain cases, i.e., the teacher is already better than any student. 
In principle, there may exist an optimal $\bT$ that is class independent for the given dataset, 
but it is rather hard to comment about its behaviour in general without explicitly accounting for the individual class-specific accuracies $\epsilon_c$s.
This is simply because correcting for noisy labels in one class may have the outcome of corrupting another class. 
Lastly, it should be apparent that the gain in performance per class $c$ has diminishing returns as $T$ increases. 

\noindent \paragraph{Choosing $\bT$ and $\alpha_t$s:} \label{sec:choose-bt-alphas}
Corollary \ref{cor:optimal-T} is an existence result and does not give us a procedure to compute $\bT_c$s (or $\bT$) and the corresponding $\alpha$s. In practice, there is a simple strategy one can follow. 
At stage $0$, we train $\N^0$ and we record its average across-class performance $\epsilon^0 = \frac{1}{C}\sum_{c=1}^C\epsilon^0_c$.
At stage $1$, we empirically select the best $\alpha_0$ that results in maximal $\epsilon^1$. 
If $\epsilon^1 \leq \epsilon^0$, then we stop and declare $\bT = 0$ i.e., no student needed. 
On the other hand, if $\epsilon^1 \geq \epsilon^0$, then we continue to stage $t=2$. And repeat this process until the accuracy $\epsilon^t$ saturates or starts to decrease. 
This averages out the per-class influence on $\bT$. 

%%%%%%%%%%%%%%%%%%%%%%%%%%%%%%%% CNN for Audio %%%%%%%%%%%%%%%%%%%%%%%%%%%%%%%%%%%%%%%%

\section{CNN for Weakly Labeled SER} \label{sec:weakcnn}

We now evaluate Algorithm \ref{alg:secost} for weakly labeled SER, as motivated in Section \ref{sec:intro}. 
We first propose a novel architecture for the problem and then study SUSTAIN using this network. 
Observe that most of the existing approaches to SER are variants of Multiple Instance Learning (MIL)~\cite{dietterich1997solving}, 
the first proposed framework being \cite{kumar2016audio}. 

Our key novelty is to include a class-specific ``attention'' learning mechanism within the MIL framework. 
We introduce some brief notation followed by presenting the model.  
In MIL, the training data $\D$ is made available via \emph{Bags} $\mathcal{B}_i$, 
with each bag corresponding to a collection of $m_i$ training instances $\{\x^1_i,\dots,\x^{m_i}_i\}$. 
Each $\mathcal{B}_i$ has one label vector $\z_i$. 
$\z_{i,c} = 1$ for class $c$ if {\it at least} one of the $m_i$ instances is positive, otherwise $\z_{i,c} = 0$. 

The key idea in MIL is that the learner first predicts on instances, 
and then maps (accumulates) these instance-level predictions to a bag-level prediction. 
For instance, a widely used SVM based MIL~\cite{andrews2002support} uses this principle, 
using $max$ operator as the mapping function. 
Based on similar principle, we formulate the learning process as follows: 
\begin{equation} \label{eq:trainloss}
    %\mathcal{L}(\Theta, \Phi) = 
    \sum_{i=1}^N \ell(g_{\Phi}(f_\Theta(\x^1_i), \cdots, f_\Theta(\x^{m_i}_i)), \z_i)
\end{equation}
$f(\cdot)$, parameterized by $\Theta$, is the learner and does the instance level prediction of outputs, 
and $g(\cdot)$ maps these $f_{\Theta}(\x^s_i)$ to bag level predictions. 
For weakly labeled SER, $\mathcal{B}_i$s are full audio recordings and instances are short duration segments of the recordings. 
We design a CNN which takes in Log-scaled Melfilter-bank feature representations of the entire audio recording, 
produces instance (i.e., segment) level predictions which are then mapped to recording level predictions. 

The inputs are computed as follows: $64$ Mel-filter-bank is obtained for each $16$ms window of audio, and the window moves by $10$ms, 
leading to $100$ Logmel frames per second of audio (with a sampling rate of $16$KHz for the audio recordings). 

\begin{table}[t]
  \centering
  \resizebox{1.0\columnwidth}{!}{
    \begin{tabular}{c|c|c}
    \toprule
    \textbf{Stage}              & \textbf{Layers}                                                           &  \textbf{Output Size}   \\
    \midrule
    Input                       & Unless specified -- (S)tride = 1, (P)adding = 1 & $1 \times 1024 \times 64$ \\
    \midrule
    \multirow{3}{*}{Block B1}   & Conv: 64, $3 \times 3$      & $64 \times 1024 \times 64$ \\
                                & Conv: 64, $3 \times 3$     & $64 \times 1024 \times 64$\\
                                & Pool: $4 \times 4$ (S:4)                                                  & $64 \times 256 \times 16$ \\
    \midrule
	\multirow{3}{*}{Block B2}   & Conv: 128, $3 \times 3$     & $128 \times 256 \times 16$ \\
                                & Conv: 128, $3 \times 3$     & $128 \times 256 \times 16$ \\
                                & Pool: $2 \times 2$ (S:2)                                                  & $128 \times 128 \times 8$ \\ 
    \midrule
	\multirow{3}{*}{Block B3}   & Conv: 256, $3 \times 3$    & $256 \times 128 \times 8$ \\
                                & Conv: 256, $3 \times 3$     & $256 \times 128 \times 8$ \\
                                & Pool: $2 \times 2$ (S:2)                                                  & $256 \times 64 \times 4 $ \\       
    \midrule
    \multirow{3}{*}{Block B4}   & Conv: 512, $3 \times 3$     & $512\times 64 \times 4 $ \\
                                & Conv: 512, $3 \times 3$     & $512 \times 64 \times 4 $ \\
                                & Pool: $2 \times 2$ (S:2)                                                  & $512 \times 32 \times 2 $ \\ 
    \midrule
    Block B5                    & Conv: 2048, $3 \times 2$ (P:0)   & $2048 \times 30 \times 1 $ \\
	\midrule 
	Block B6                    & Conv: 1024, $1 \times 1$   & $1024 \times 30 \times 1 $ \\
	\midrule
	Block B7                   & Conv: 1024, $1 \times 1$ & $1024 \times 30 \times 1 $ \\
	\midrule
	Block B8                    & Conv: C, $1 \times 1$  & $C \times 30 \times 1 $ \\
	\midrule
	$g()$	                        & $\mathbf{W}_{\Phi}$ ($C \times C$)       & $C \times 1 $  \\
	\bottomrule      
    \end{tabular}
    }
  \caption{WEANET: All convolutional layers (except B8) are followed by batch norm and ReLU; Sigmoid activation follows B8.}%
  \label{tab:cnnarch}
\end{table}%

The proposed architecture, referred to as WEakly labeled Attention NETwork (WEANET), is shown in Table \ref{tab:cnnarch}. 
Example output sizes for an input with $1024$ Logmel frames (approx. 10 seconds long audio) is shown in Table \ref{tab:cnnarch}. 
The first few convolutional layers (B1 to B5) produce $2048$ dimensional bag representations for the input at Block B5. 
B6-B8 are $1 \times 1$ convolutional layers that produce instance level predictions of size $C \times K$, $C$ is number of classes and $K$ is number of segments obtained for a given input. 
The network is designed such that the receptive field of each segment (i.e., instance) is $\sim 1$ second ($96$ frames), and the segments themselves move by $0.33$ seconds ($32$ frames). 
The instance level predictions are then used to produce bag (i.e., recording) level predictions using $g(\cdot)$. 
An easy parameter free way of doing this is to use $mean$ (or $max$) functions which will simply take average (or maximum) over segment level predictions from B8. 

Instead, we propose an attention mechanism here which aims to appropriately weigh each segment's contribution in the final recording level prediction. 
Moreover, this is done in a class-specific manner as different sounds might be located at different places in the recording. 
More formally, $g(\cdot)$ is parameterized as follows: 
\begin{align} \label{eq:gmap}
\mathbf{A} &= \tilde{\sigma}(\mathbf{W}_{\Phi} \mathbf{S}) \\
\mathbf{o} &= \sum_{k=1}^K \tilde{\mathbf{O}}_k \quad\text{s.t.}\quad \tilde{\mathbf{O}} = \mathbf{A} \odot \mathbf{S} 
\end{align}
% POINT OUT  THAT W does not grow with size of input
where $\mathbf{W}_{\Phi} \in \mathrm{R}^{C\times C}$. $\mathbf{S} \in \mathrm{R}^{C\times K}$ denotes the segment level predictions. 
$\tilde{\sigma}$ is the softmax function applied across segments, 
and $\mathbf{A} = \tilde{\sigma}(\mathbf{W}_{\Phi} \mathbf{S})$ gives us the attention weights for each segment and class.
$\odot$ is element wise multiplication and $\tilde{\mathbf{O}}_k$ is $k^{th}$ column of $\tilde{\mathbf{O}}$, 
which represents the {\it weighted} predictions for each class in $k^{th}$ segment. 
All these are then pooled into $\mathbf{o}$, which represents the recording level prediction for the input.  
$\mathbf{W}_{\Phi}$ is learned along with rest of the parameters of the WEANET. Note that, the size of attention parameter $\mathbf{W}_{\Phi}$ is independent of the number of segments obtained for an input or in other words the duration of the input. It depends on the number of classes in the dataset. 

%%%%%%%%%%%%%%%%%%%%%%%%%%%%%%%%% Experiments %%%%%%%%%%%%%%%%%%%%%%%%%%%%%%%%%%%%%%%%%

\section{Experiments and Results} \label{sec:expt}

\subsection{Datasets and Experimental Setup}

\noindent {\bf Audioset:}~\cite{gemmeke2017audio} \label{sec:audioset} is very challenging dataset in terms of adverse learning conditions outlined in Section \ref{sec:intro}. It is the largest dataset for sound events with weakly labeled YouTube clips for $527$ sound classes. 
Each recording is $\sim 10$ seconds long and on an average, there are $2.7$ labels per recording. 
%Audioset comes into pre-defined training set and evaluation set. 
The training and evaluation sets consist of $\sim 2$ million and $\sim 20,000$ recordings respectively. 
The dataset is highly unbalanced with  the number of training examples varying from close to $1$ million for classes such as \emph{Music} and \emph{Speech} to $<100$ for classes such as \emph{Screech} and \emph{Toothbrush}. 
The evaluation set has at least $59$ examples for each class. 
A sample of $\sim 25,000$ videos from the training set are sampled out for validation. 

An analysis of label noise was done by the authors by sampling 10 examples for each class and sending them for expert label reviewing. This puts label noise at broad range of 0 to 80-90\%  across classes. Note however that this is an extremely rough estimate for a dataset of this size. 

\noindent {\bf FSDKaggle:}\label{sec:fsdkaggle} ~\cite{fonseca2019audio} is a dataset of 80 sound events. It has $2$ training sets: a \emph{Curated} set with $4970$ recordings and a \emph{Noisy} set with $19,815$ audio recordings. 
The \emph{Curated} set is a clean training set which has been carefully annotated by humans to ensure minimal to no label noise. 
The \emph{Noisy} training set is obtained from Flickr videos and not labeled by humans. They contain considerable amount of label noise. The evaluation set has $3361$ recordings. 
We use the \emph{Public} test set with $1120$ recordings for validation. 

~\cite{fonseca2019audio} does a more thorough examination of label noise. The estimated per-class label noise roughly ranges from 20\% to 80\% and overall around  60\% of the labels show some type of label noise. While the \emph{Curated}, validation and test sets are sourced from freesounds.org (and then labeled by humans), 
the \emph{Noisy} training set recordings are sourced from Flickr. This heavy mismatch in domain adds on to the already difficult learning conditions for the \emph{Noisy} training set and leads to considerable impact on performance.  

\noindent {\bf ESC-50:}~\cite{piczak2015esc} \label{sec:esc50}
This dataset consists of $2000$ recordings from $50$ sound classes. 
Each sound class has 40 audio recordings and all recordings are $5$ seconds long. 
We use this dataset primarily in our transfer learning experiments in Section \ref{sec:transfer}. 
It comes with $5$ pre-defined sets and we follow the same setup in our experiments as in prior works such as \cite{kumar2017knowledge}. 

\noindent {\bf Experimental Setup:} \label{sec:exp-setup}
All of our experiments uses Pytorch~\cite{paszke2017automatic} for neural network implementations. Adam optimizer is used, and networks are trained for 20 epochs. Minibatch size is set to 144. Hyperparameters such as learning rates and the best model during training is selected using the validation set. The attention weight parameter $\mathbf{W}_{\Phi}$ is initialized with $0$'s such that the initial attention weights come out to be equal for all segments for all classes, $\mathbf{A}_{ck} = 1/K$, $\forall c, k$. The updates for attention weight parameter is turned on from fifth epoch. For Audioset, given its highly unbalanced nature, we use a weighted loss for each class. 
This weight for class $c$ is given by $w_c = 1 + log_2(\gamma_c)$, 
where $\gamma_c$ is the inverse of the class prior in the training set. The training set up is consistent across all stages of SUSTAIN and only the teacher(s) and the parameter $\alpha$ changes. 

Similar to prior works on Audioset, Average Precision (AP) and Area under ROC curves (AUC) are used to measure performance. 
Mean AP (mAP) and mean AUC over all classes are used as overall metrics for performance assessment. For FSDkaggle dataset, the metric used is a label-weighted label-ranking average precision (lwlrap)~\cite{fonseca2019audio}. 
Given the smaller size of this dataset, we use a lighter version of WEANET. 
The details of this lighter WEANET are provided in the supplementary material. For ESC-50 dataset accuracy is used as the performance metric.

\subsection{WEANET Model} \label{sec:weanet-eval}

\begin{table}[t!]
  \centering
    %\resizebox{1.0\columnwidth}{!}{
    \begin{tabular}{c|c|c}
    \toprule
	Method & mAP & mAUC \\ %& Method & mAP & mAUC \\
    \midrule
    \cite{kong2019weakly}-1 & 0.361 & 0.969 \\
	\cite{wang2019comparison}-1  & 0.354 & 0.963 \\ 
	\midrule
	\cite{kong2019weakly}-2   & 0.369 &  0.969 \\
	\cite{wang2019comparison}-2 & 0.362  & 0.965 \\
	\midrule
	WEANET ($g() = avg()$) & 0.352 & 0.970 \\
	WEANET ($g(\dot,\mathbf{W}_\Phi)$)  & 0.366  & 0.958 \\
	\bottomrule
    \end{tabular}%}
    \caption{Comparison of WEANET with other attention architectures on Audioset dataset.}
  \label{tab:weanet}%
\end{table}%

\begin{figure}[t]
      \centering
      \includegraphics[width=\linewidth]{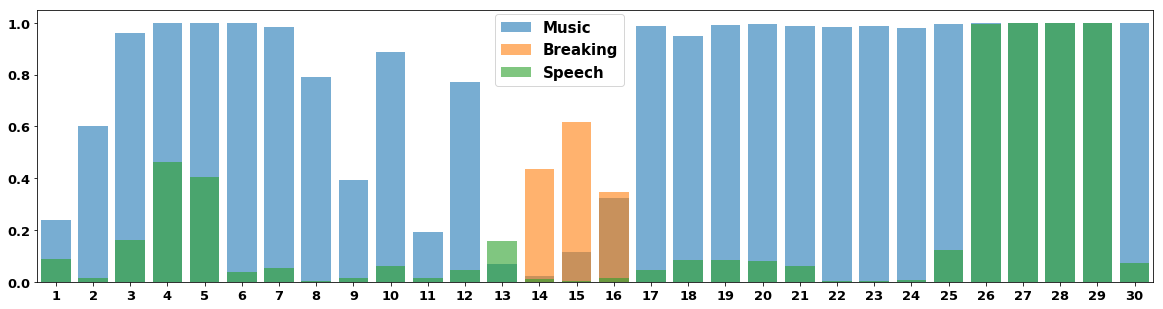}
      \includegraphics[width=\linewidth]{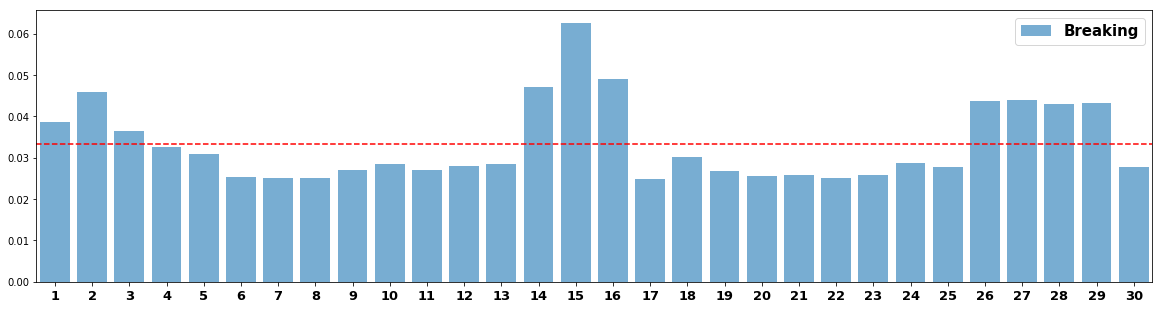}
      \includegraphics[width=\linewidth]{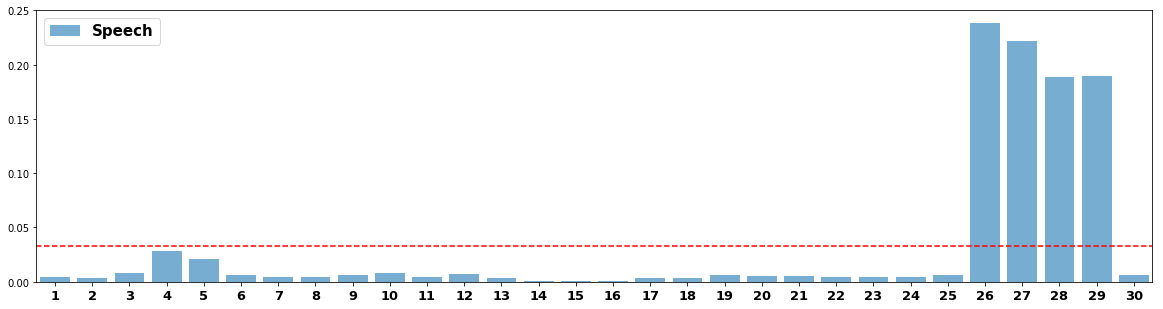}
    \caption{An example of WEANET outputs on a recording from test set. \textbf{Top}: Segment level probability outputs for three classes present in the recording. \textbf{Mid}: Segment wise attention weights ($\mathbf{A}$) for the Breaking Sound. \textbf{Bottom}: Segment wise attention weights ($\mathbf{A}$) for the Speech Sound. Red line denotes if all segments were given equal weights (1/30).}
    \label{fig:example_out}
\end{figure}

We first provide some results on WEANET. 
Table \ref{tab:weanet} shows performance of WEANET framework and compares it with respect to some other attention frameworks for weakly labeled SER. 
Note that, \cite{kong2019weakly} uses embeddings for audio recordings from a network trained on a very large database (YouTube-70M) \cite{hershey2017cnn}. 
These pre-trained representations lead to enhanced performance on Audioset. 
We (and also \cite{wang2018comparing}) work with the actual audio recordings and use logmel feature representations.  
In summary, WEANET performs better than other attention frameworks. 
\cite{kong2019weakly}-2 performs slightly better but uses pre-trained embeddings as just mentioned. 
WEANET with class-specific attention is $~4\%$ better than WEANET with $g(\cdot)$ as simple average pooling. 

The major advantage of having class-specific attention learning is for localization of events. 
Figure \ref{fig:example_out} shows segment level outputs for $3$ sounds present in a specific recording from the test set. 
Note that we matched with the location of the events in the actual recording. 
The lower two figures show attention weights for the two events (\emph{Breaking} sound and \emph{Speech} sound) that are highly localized in the recording. 
Observe that the weights are much higher than average for segments where the event is actually located. 
For speech in particular, segments $1-6$ show high probability of presence even though actually speech is not present. 
However, the class-specific attention framework is capable of flagging this false positive and assigns very low weights to them.   

\subsection{SUSTAIN Framework} \label{sec:sustain-eval}

\paragraph{Comparison with state-of-the-art:}
\begin{table}[t!]
  \centering
    \resizebox{\columnwidth}{!}{
    \begin{tabular}{c|c|c}
    \toprule
	Method & mAP & mAUC \\
    \toprule
    \cite{kong2019weakly} - Small & 0.361 & 0.969 \\
    \cite{kong2019weakly} - Large   & 0.369 &  0.969 \\
	\midrule 
	\cite{wang2019comparison} - TALNet (exp. pooling)  & 0.362   & 0.965 \\ 
	\cite{wang2019comparison} - TALNeT (Attention) & 0.354  & 0.963 \\
	\midrule
	\cite{ford2019deep} - ResNet-34 (Attention) & 0.360 & 0.966 \\
	\cite{ford2019deep} - ResNet-101 (Attention) & 0.380 & 0.970 \\
	\midrule
	WEANET & 0.366 & 0.958 \\
	\textbf{SUSTAIN - Single Teacher} & \textbf{0.394} & \textbf{0.972} \\
	\textbf{SUSTAIN - 2 Teachers} & \textbf{0.398} & \textbf{0.972} \\
	
	\bottomrule
    \end{tabular}}
    \caption{Comparison with state-of-the-art methods on Audioset}
  \label{tab:sota}%
\end{table}%

Table \ref{tab:sota} compares performance of our SUSTAIN framework with state-of-the-art methods on Audioset. 
``SUSTAIN-Single Teacher" uses $1$ teacher at each stage, specifically the network trained in the previous stage (as in Eq. \ref{eq:new-labels-1-teacher}). 
``SUSTAIN-2 Teachers" uses $2$ networks learned in the previous $2$ consecutive stages as teachers.
SUSTAIN learning leads to superior performance over all prior methods. 
The ResNet based architectures in Table \ref{tab:sota} are much 
larger compared to our WEANET and have several times more parameters. 
Our method outperforms the previous best method (ResNet-50) by $~4.7\%$. Note that, \cite{ford2019deep} also reports a performance of $0.392$ but that is obtained through ensemble of models, by averaging outputs of multiple models.  
%\begin{wrapfigure}[11]{r}{0.57\linewidth} 
% \includegraphics[width=\linewidth]{figures/1teacher_1stage_vs_alphas.png}
% \caption{\label{fig:vs_alphas} Single teacher: $\N^1$ vs. $\alpha_0$}
%\end{wrapfigure}

\begin{figure}[t] 
\centering
 \includegraphics[width=0.8\columnwidth]{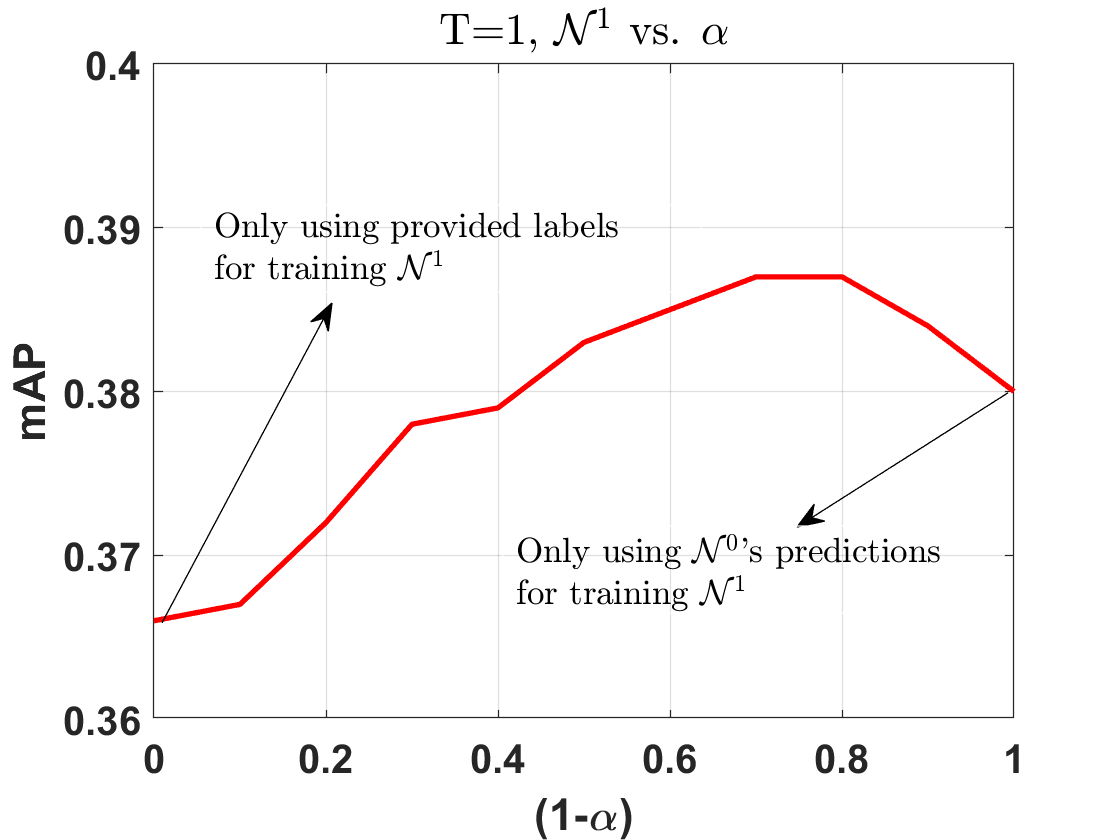}
 \caption{\label{fig:vs_alphas} Single teacher: $\N^1$ vs. $\alpha_0$}
\end{figure}

Focusing primarily on the SUSTAIN learning, we notice that it can lead up to $8-9\%$ improvement in results for the WEANET model. Thus, the same architecture WEANET, generalizes much better after a few stages of SUSTAIN learning as opposed to just training it on the available labels.

\noindent {\bf Single Stage ($T=1$) vs. varying $\alpha_0$:}
Figure \ref{fig:vs_alphas} shows that $\alpha_0$ influences mAP, as suggested by in Section \ref{sec:remarks}. %
The two extremes of $\alpha_0=0$ (only using $\N^0$'s predicted labels) and $\alpha_0=1$ (only using provided labels $\y^s$ for learning) perform worse than learning using a combination of the two. 
This asserts our primary claim in Proposition \ref{prop:delta-bar}. Depending on the weight ($1 - \alpha_0$) given to the teacher, even a single stage of SUSTAIN can lead to up to 5.7\% improvement in performance. 

\noindent {\bf Multiple Stages ($T>1$) vs. $\alpha_0$:}
For a fixed $\alpha_0$, Figure \ref{fig:Nt_alpha}(a) shows that as $T$ increases mAP starts to increase and then quickly saturates, showing evidence for Corollary \ref{cor:optimal-T}. 
Note that this setup corresponds to using the last trained network as teacher (i.e., $\N^3$ uses $\N^2$ as the teacher).
It is reasonable to expect that $\N^t$ is better than $\N^{t-1}$, and so, 
one can put more confidence in the predicted labels of latest teachers than teachers from earlier stages. 
This is validated in Figure \ref{fig:Nt_alpha}(b) where $\N^0$ uses only predicted labels ($\alpha_0=1$) and as $T$ increases, we reduce $\alpha_0$, 
putting more confidence on teacher's predictions and achieve better mAP.  Unlike the fixed alpha case, considerable improvement is obtained from stage 1 to 2 and then stage 2 to 3 by increasing the weight given to teacher's predictions. 

\begin{figure}[t!]
      \centering
      \includegraphics[width=0.51\linewidth]{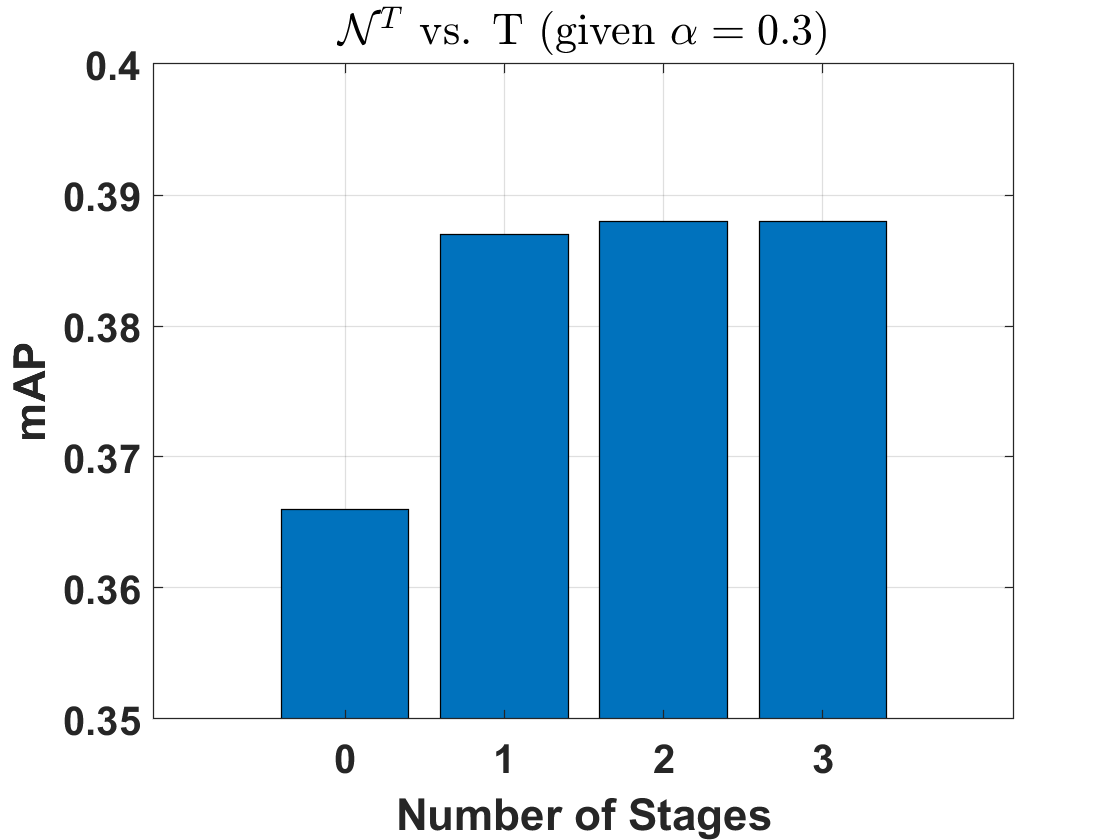} \hspace{-5mm}
      \includegraphics[width=0.51\linewidth]{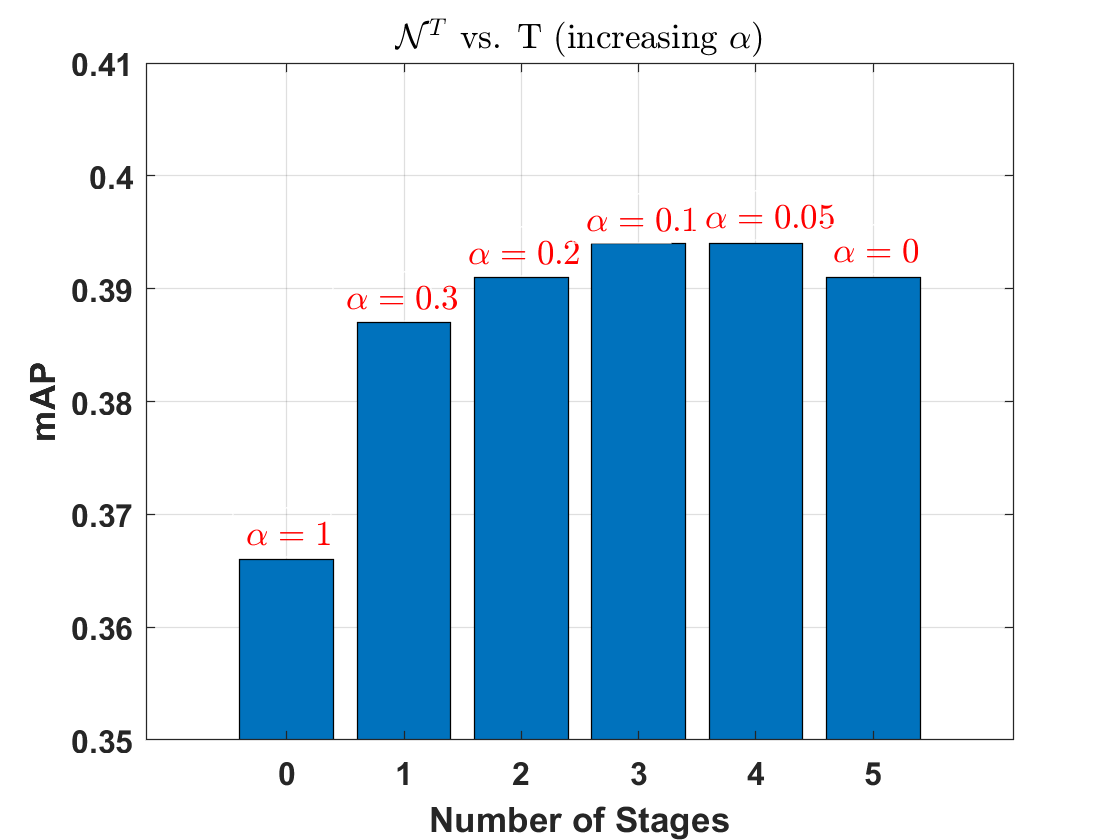}
    \caption{Performance of students as $T$ increases: (a) $\alpha_0=0.3$, (b) decreasing $\alpha_0$ as $T$ increases.}%
    \label{fig:Nt_alpha}
\end{figure}

\begin{table}[t!]
  \centering
    \resizebox{1.0\columnwidth}{!}{
    \begin{tabular}{|c|c|c|c|c|c|c|c}
    \toprule
    Stage(T) & Teach.  & $\alpha_1$, $\alpha_2$ & $\alpha_0$  & Stud. & Stud. Perf.  \\
    \midrule
    0 & - & - & 1.0 & $\mathcal{N}^0$ & 0.366 \\
    \midrule
    1 & $\mathcal{N}^0$, - & 0.7, - & 0.3 & $\mathcal{N}^1$ & 0.387 \\ 
    \midrule
    2 & $\mathcal{N}^0$, $\mathcal{N}^1$ & 0.3, 0.5 & 0.2 & $\mathcal{N}^2$ & 0.393 \\
    \midrule
    3 & $\mathcal{N}^1$, $\mathcal{N}^2$ & 0.4, 0.5 & 0.1 & $\mathcal{N}^3$ & 0.396 \\
   \midrule
   4 & $\mathcal{N}^2$, $\mathcal{N}^3$ & 0.45, 0.5 & 0.05 & $\mathcal{N}^4$ & 0.398 \\
   \midrule
   5 & $\mathcal{N}^3$, $\mathcal{N}^4$ & 0.45, 0.53 & 0.03 & $\mathcal{N}^5$ & 0.398 \\
   \bottomrule
    \end{tabular}%
    }
	\caption{$2$ teachers at each stage (weights: $\alpha_1$ and $\alpha_2$).} %$\N^0$ is the default teacher.}  
  \label{tab:multipleteach}%
\end{table}%

\noindent {\bf Multiple Stages and Multiple Teachers:}
Table \ref{tab:multipleteach} shows the performance as $T$ increases with $2$ teachers. 
Each row corresponds to one stage and $\N^0$ is as usual the default teacher. 
As expected, the mAP improves as $T$ increases. 
In particular, observe that after $5$ stages we reach $0.398$ mAP here, compared to the single teacher setup where we get $0.392$ after $5$ stages (refer to Figure \ref{fig:Nt_alpha}(b)), a $1.5\%$ relative improvement. 
We also see the saturation of performance as $T$ increases, further supporting our results from Section \ref{sec:analysis}.

\begin{table}[t!]
  \centering
    %\resizebox{\columnwidth}{!}{
    \begin{tabular}{c|c|c}
    \toprule
	Method & lwlrap & Remarks\\
    \toprule
    \cite{fonseca2019audio} & 0.312 & -\\
	\midrule 
	WEANET$^L$ - (T = 0) & 0.436 & -\\ 
	\midrule
	SUSTAIN - (T = 1) & 0.454 & $\alpha = 0.5$ \\
	SUSTAIN - (T = 2) & 0.456 & $\alpha = 0.3$ \\
	SUSTAIN - (T = 3) & 0.462 & $\alpha = 0.2$ \\
	SUSTAIN - (T = 4) & 0.470 & $\alpha = 0.15$ \\
	SUSTAIN - (T = 5) & 0.472 & $\alpha = 0.05$ \\
	\textbf{SUSTAIN - (T = 6)} & \textbf{0.472} & $\alpha = 0.05$ \\
	
	\bottomrule
    \end{tabular}%}
    \caption{Performance when trained on FSDKaggle-Noisy set. First row is baseline. Last column shows $\alpha$ for each stage. Single Teacher ($\mathcal{N}^{T-1}$)  at each  stage of training. }
  \label{tab:fsdnoisy}%
  %\vspace{-0.1in}
\end{table}%

\subsection{Noisy label vs Clean Label Conditions:}

We now try to specifically look into clean and noisy label learning conditions using FSDKaggle2019 dataset. 
The \emph{Noisy} and \emph{Curated} training sets of this dataset (refer to their descriptions from Section \ref{sec:exp-setup}) are used as noisy (i.e., hard) and clean (i.e., easy) learning conditions. 
The test set remains same for the two cases. 
We use a lighter version of WEANET model (WEANET$^L$) for these experiments, the details of which are available in the supplementary material. For all these experiments, only one teacher is used per stage; the network trained in the previous stage. Most prior works on FSDKaggle have relied heavily on different forms of data augmentation on the \emph{Curated} set for improved performance. We do not do any data augmentation and instead focus  on easy and hard conditions. We use the performance reported by the dataset paper, \cite{fonseca2019learning}, on each training set as the baseline. 

Table \ref{tab:fsdnoisy} summarizes the results for \emph{Noisy} training set. We observed that for the \emph{Noisy} training set, the trends of results (for different parameters such as $\alpha$s, stages etc.) are similar to those observed for Audioset.  Overall SUSTAIN leads to around $8\%$ improvement. Even just $1$ stage of SUSTAIN, leads to almost $4.1\%$ improvement in performance over base WEANET$^L$ (T = 0). The performance saturates after 5 stages of SUSTAIN.

The \emph{Curated} set presents a different picture though. While one stage of SUSTAIN still leads to small improvement, any further co-supervision leads to deterioration in performance. This is the expected behavior our technical results claimed in Section \ref{sec:analysis} i.e., SUSTAIN learning primarily helps in adverse learning conditions.

\begin{table}[t!]
  \centering
    \resizebox{0.95\columnwidth}{!}{
    \begin{tabular}{c|c|c}
    \toprule
	Method & lwlrap & Remarks\\
    \toprule
    Baseline - \cite{fonseca2019audio} & 0.542 & -\\
	\midrule 
	WEANET$^L$ - (T = 0) & 0.619 & -\\ 
	\midrule
	SUSTAIN - (T = 1) & 0.619 & $\alpha = 0.9$ \\
	SUSTAIN - (T = 1) & 0.625 & $\alpha = 0.7$ \\
	SUSTAIN - (T = 1) & 0.632 & $\alpha = 0.5$ \\
	SUSTAIN - (T = 1) & 0.622 & $\alpha = 0.3$ \\
	SUSTAIN - (T = 1) & 0.622 & $\alpha = 0.1$ \\
	\midrule
	SUSTAIN - (T = 2) & 0.624 & $\alpha = \{0.5,0.9\}$ \\
	SUSTAIN - (T = 2) & 0.623 & $\alpha = \{0.5,0.7\}$ \\
	SUSTAIN - (T = 2) & 0.627 & $\alpha = \{0.5,0.5\}$ \\
	SUSTAIN - (T = 2) & 0.624 & $\alpha = \{0.5,0.3\}$ \\
	SUSTAIN - (T = 2) & 0.625 & $\alpha = \{0.5,0.1\}$ \\
	
	\bottomrule
    \end{tabular}}
    \caption{Performance when trained on FSDKaggle-Curated set. Last row last column shows $\alpha$ for each stage. }
  \label{tab:fsdclean}%
\end{table}%

\subsection{Class-specific Performance Gains}
On the Audioset dataset, we observed that for almost $85\%$ of all the classes (527 total), the performance improved with SUSTAIN learned model $\mathcal{N}^4$ (from Table \ref{tab:multipleteach}), compared to base WEANET model ($\mathcal{N}^0$ from Table \ref{tab:multipleteach}). 
Most classes have under $25\%$ relative improvement, and $69$ of the classes get $>25\%$ improvement, and this reaches up to $100\%$ for classes like \emph{Squeal} and \emph{Rattle}. 
Maximum drop in performance (down by $30\%$) is observed for \emph{Gurgling} class. 
We also see that low performing classes (AP $< 0.1$) have more improvements in relative sense. 
On average, AP of these classes ($44$ of them) improve by $23\%$, while classes with high AP ($> 0.5$, $146$ in number), we see $6\%$ gain in performance. Class-specific performance plots are shown in supplementary material. Overall, \emph{Bagpipes} sounds are easiest to recognize and we achieve an AP of 0.931 for it. \emph{Squish} on the other hand is hardest to recognize with an AP of 0.02.

\section{Knowledge Transfer using SUSTAIN} \label{sec:transfer}

In the preceding sections, we showed that the generalizability of a model can be improved through the proposed SUSTAIN learning. We now ask, are the models obtained from SUSTAIN learning more suitable for transfer learning? 
We study whether WEANET obtained after $T$ stages of training ($\mathcal{N}^T$) is more suited for transfer learning compared to the one just trained on the available labels.
Since SUSTAIN is not explicitly designed to handle this, the transfer learning question reveals the learnability power of the proposed framework. 

We pick $\mathcal{N}^4$ and $\mathcal{N}^0$ WEANET models from Table \ref{tab:multipleteach} for this analysis, with  $\mathcal{N}^0$ being the base model trained only on available labels and $\mathcal{N}^4$ being a SUSTAIN trained model. These WEANET models trained on Audioset are used to obtain representations for the audio recordings in the given target tasks. 
Outputs after Block B5 (refer to WEANET model from Table \ref{tab:weanet}) are used as feature representations for the audio recordings. 
Recall that, Block B5 produces representations for $1$ second long audio every $0.33$ sec. These segment level representations are simply max-pooled across all segments to get a fixed $2048$-dimensional vector for all audio recordings. 

We study these transfer learning tasks on FSDKaggle and ESC-50 datasets. A simple linear classifier is trained on the feature representations obtained for the audio recordings. 

Table \ref{tab:transfer} shows the results for these transfer learning tasks. We see that the representations from SUSTAIN framework leads to significantly improved feature learning for all datasets. For the clean conditions (ESC-50 and FSDKaggle-Curated), we see 1.5-2.2\% improvement whereas for the noisy learning conditions we see up to 3.5\% improvement in performance.  For the ESC-50 dataset, this transfer learning also outperforms previous state-of-the-art results by a considerable margin (2.8\% relative). 

\begin{table}[t!]
  \centering
    \resizebox{\columnwidth}{!}{
    \begin{tabular}{cc|c|c}
    \toprule
    \multicolumn{2}{c|}{ESC-50} & \multicolumn{2}{c}{FSDKaggle}\\
    \midrule
    Method & Acc. (\%) & Method & lwlrap \\
    \midrule
    \cite{sailor2017unsupervised} & 86.5 & Noisy, WEANET ($\mathcal{N}^0$) & 0.486\\
    \cite{guzhov2020esresnet} & 91.5 & Noisy, WEANET ($\mathcal{N}^4$) & 0.503\\
    \cmidrule{3-4}
    WEANET ($\mathcal{N}^0$) & 92.6 & Curated, WEANET ($\mathcal{N}^0$) & 0.712 \\
    WEANET ($\mathcal{N}^4$) & 94.1 & Curated, WEANET ($\mathcal{N}^4$) & 0.728\\
	\bottomrule
    \end{tabular}}
    \caption{Transfer Learning from SUSTAIN Models trained on Audioset. Results on ESC-50 and FSDKaggle dataset.}
  \label{tab:transfer}%
\end{table}%

%%%%%%%%%%%%%%%%%%%%%%%%%%%%%%%%%%%%%%%%%%%%%%%%%%%%%%%%%%%%%%%%%%%%%%%%%%%%%%%

\section{Conclusions} \label{sec:conc}
%Weakly labeled datasets learning poses interesting challenges. 

Designing robust learning models for weakly labelled datasets while also scaling them to large scale and ensuring good generalization is a hard problem, and is an open question. 
We addressed this problem in this paper. We proposed a sequential self-teaching framework that utilizes co-supervision across trained models to improve generalization. 
We specifically show promising results on sound event recognition and detection, in particular in large scale weakly labelled settings. We also proposed a novel architecture for learning sounds which incorporates class-specific attention learning. A better theoretical understanding of the role $\alpha$ and $T$ play in different adverse learning conditions can lead to an enhanced understanding of SUSTAIN. We will explore these directions in future works. 

\bibliography{references}
\bibliographystyle{icml2020}

%%%%%%%%%%%%%%%%%%%%%%%%%%%%%%%%%%%%%%%%%%%%%%%%%%%%%%%%%%%%%%%%%%%%%%%%%%%%%%%

%\pagebreak
%\widetext
\clearpage
%%%%%%%%%% Merge with supplemental materials %%%%%%%%%%
%%%%%%%%%% Prefix a "S" to all equations, figures, tables and reset the counter %%%%%%%%%%
\setcounter{equation}{0}
\setcounter{figure}{0}
\setcounter{table}{0}
\setcounter{section}{0}
\setcounter{page}{1}
\makeatletter
\renewcommand{\theequation}{S\arabic{equation}}
\renewcommand{\thefigure}{S\arabic{figure}}
\renewcommand{\bibnumfmt}[1]{[S#1]}
\renewcommand{\citenumfont}[1]{S#1}
\renewcommand{\thetable}{S\arabic{table}}

\begin{center}
\textbf{\large Supplementary Materials}
\end{center}

\section{Technical Results}

\begin{align} \label{seq:crossent}
\L(\p^s, \y^s) &= \frac{1}{C}\sum_{c=1}^{C} \l(\p^s, \y^s) \hspace{2mm} \text{where} \\ 
\l(p^s_c, y^s_c) &= - y^s_c \log(p^s_c) - (1-y^s_c) \log(1-p^s_c)
\end{align}

\begin{equation} \label{seq:new-labels}
\by_t^s = \alpha_0 \y^s + \sum_{\tt=1}^t \alpha_\tt \hp_{\tt-1}^s \quad \text{\it s.t.} \quad \sum_{\tt=0}^{t} \alpha_\tt = 1 
\end{equation}

\begin{equation} \label{seq:new-labels-1-teacher}
\by_t^s = \alpha_0 \y^s + (1 - \alpha_0) \hp_{t-1}^s    
\end{equation}

\begin{equation} \label{seq:true-labels}
y^s_c = \begin{cases} 
y^{*s}_c & \text{\it w.p.} \hspace{2mm} \delta_c \\ 
1-y^{*s}_c & \text{\it else} \end{cases}
\end{equation}

\begin{equation} \label{seq:new-label-1-stage}
\by_1^s = \alpha_0 \y^s + (1-\alpha_0) \hp_0^s \\
\end{equation}

\begin{equation} \label{seq:one-teacher-delta-bar-def}
\hat{p}^s_{0,c} = \begin{cases} 
y^{*s}_c & \text{\it w.p.} \hspace{2mm} \bar{\delta_c} \\ 
1-y^{*s}_c & \text{\it else} \end{cases}
\end{equation}  

\begin{proposition} \label{sprop:delta-bar}
  Let $\N^1$ be trained using $\{\x^s, \by^s\} \;\forall \;s$ using binary cross-entropy loss, and let $\epsilon_c$ denote the average accuracy of $\N^0$ for class $c$.
  Then, we have
  \begin{equation} \label{seq:one-teacher-delta-bar-val}
    \bar{\delta}_c = \epsilon_c\delta + (1-\epsilon_c)(1-\delta) \; \forall \; c
  \end{equation}  
  and whenever $\delta < \frac{1}{2}$, $\N^1$ improves performance over $\N^0$. The per class performance gain is $(1-\epsilon_c)(1-2\delta)$ 
\end{proposition}
\begin{proof}
Recall the entropy loss from Eq. \ref{seq:crossent}, for a given $s$ and $c$. 
Using the definition of the new label from Eq. \ref{seq:new-label-1-stage}, we get the following 
\begin{equation} \label{seq:prop-1}
    \l(p^s_c, \bar{y}^s_c) = \alpha_0 \l(p^s_c, y^s_c) + (1-\alpha_0) \l(p^s_c, \hat{p}^s_c)  
\end{equation}
Now, Eq. \ref{seq:true-labels} says that {\it w.p.} $\delta$ (recall $\delta_c = \delta \; \forall \; c$ here), $\l(p^s_c, y^s_c) = \l(p^s_c, y^{*s}_c)$, else $\l(p^s_c, y^s_c) = \l(p^s_c, 1-y^{*s}_c)$. 
Hence, using Eq. \ref{seq:true-labels} and Eq. \ref{seq:one-teacher-delta-bar-def}, and using the resulting equations in Eq. \ref{seq:prop-1} we have the following
\begin{align*}
    &\mathop{\mathbb{E}}_s \l(p^s_c, y^s_c) = \delta \sum_{s=1}^S \l(p^s_c, y^{*s}_c) + (1-\delta) \sum_{s=1}^S \l(p^s_c, 1-y^{*s}_c) \\
    &\mathop{\mathbb{E}}_s \l(p^s_c, \hat{p}^s_c) = \bar{\delta}_c \sum_{s=1}^S \l(p^s_c, y^{*s}_c) + (1-\bar{\delta}_c) \sum_{s=1}^S \l(p^s_c, 1-y^{*s}_c) \\
    &\mathop{\mathbb{E}}_s \l(p^s_c, \bar{y}^s_c) = (\alpha_0\delta + (1-\alpha_0)\bar{\delta}_c) \sum_{s=1}^S \l(p^s_c, y^{*s}_c) \\ &\quad + (\alpha_0(1-\delta) + (1-\alpha_0)(1-\bar{\delta}_c)) \sum_{s=1}^S \l(p^s_c, 1-y^{*s}_c)
\end{align*}
If $(\alpha_0\delta + (1-\alpha_0)\bar{\delta}_c) > \delta$ then we can ensure that using $\bar{y}^s_c$ as targets is better than using $y^s_c$. 
Now given the accuracy of $\N^0$ denoted by $\epsilon_c \; \forall \; c$, combining Eq. \ref{seq:true-labels} and Eq. \ref{seq:one-teacher-delta-bar-def}, we can see that $\bar{\delta}_c = \epsilon_c\delta + (1-\epsilon_c)(1-\delta)$. Using this, for $\N^1$ to be better than $\N^0$,  we need
\begin{equation} \label{seq:prop-2}
    \alpha_0\delta + (1-\alpha_0)(\epsilon_c\delta + (1-\epsilon_c)(1-\delta)) > \delta
\end{equation}
which requires $\delta < \frac{1}{2}$. And the gain is simply $\alpha_0\delta + (1-\alpha_0)\bar{\delta}_c - \delta$ which reduces to $(1-\epsilon_c)(1-2\delta)$. 
\end{proof}

\begin{corollary} \label{scor:optimal-T}
Let $\epsilon_c^t$ denote the accuracy of $\N^t$ for class $c$. 
Given some $\delta$, there exists an optimal $\bT^c$
%and a corresponding sequence of weights $\alpha_0,\ldots,\alpha_{\bT}$, 
such that $\epsilon_c^{\bT^c} \geq \epsilon_c^t$. 
\end{corollary}
\begin{proof}
When $\delta > \frac{1}{2}$, Eq. \ref{seq:prop-2} will not hold, and Proposition \ref{sprop:delta-bar} says that $\N^1$ is worse than $\N^0$. Hence $\bT^c = 1$. 
On the other hand, if $\delta < \frac{1}{2}$, then $\bar{\delta}_c > \delta$, and the performance improves. 
For the given $c$, one can repeat the analysis for next stages with different values of $\bar{\delta}_c$. 
$\bT_c$ is the stage $t$ where the corresponding $\bar{\delta}_c$ increases over $\frac{1}{2}$. 
%most probably at the extent of loss in accuracy for other classes. 
%
\end{proof}

\section{WEANET$^L$ for FSDKaggle-2019}
\begin{table}[t]
  \centering
  \resizebox{1.0\columnwidth}{!}{
    \begin{tabular}{c|c|c}
    \toprule
    \textbf{Stage}              & \textbf{Layers}                                                           &  \textbf{Output Size}   \\
    \midrule
    Input                       & Unless specified -- (S)tride = 1, (P)adding = 1 & $1 \times 1024 \times 64$ \\
    \midrule
    \multirow{3}{*}{Block B1}   & Conv: 64, $3 \times 3$      & $64 \times 1024 \times 64$ \\
                                & Conv: 64, $3 \times 3$     & $64 \times 1024 \times 64$\\
                                & Pool: $4 \times 4$ (S:4)                                                  & $64 \times 256 \times 16$ \\
    \midrule
	\multirow{3}{*}{Block B2}   & Conv: 128, $3 \times 3$     & $128 \times 256 \times 16$ \\
                                & Conv: 128, $3 \times 3$     & $128 \times 256 \times 16$ \\
                                & Pool: $2 \times 2$ (S:2)                                                  & $128 \times 128 \times 8$ \\ 
    \midrule
	\multirow{3}{*}{Block B3}   & Conv: 256, $3 \times 3$    & $256 \times 128 \times 8$ \\
                                & Conv: 256, $3 \times 3$     & $256 \times 128 \times 8$ \\
                                & Pool: $2 \times 2$ (S:2)                                                  & $256 \times 64 \times 4 $ \\       
    \midrule
    \multirow{3}{*}{Block B4}   & Conv: 256, $3 \times 3$     & $256\times 64 \times 4 $ \\
                                & Conv: 256, $3 \times 3$     & $256 \times 64 \times 4 $ \\
                                & Pool: $2 \times 2$ (S:2)                                                  & $256 \times 32 \times 2 $ \\ 
    \midrule
    Block B5                    & Conv: 512, $3 \times 2$ (P:0)   & $512 \times 30 \times 1 $ \\
	\midrule
	Block B6                    & Conv: C, $1 \times 1$  & $C \times 30 \times 1 $ \\
	\midrule
	$g()$	                        & Global Average Pooling       & $C \times 1 $  \\
	\bottomrule      
    \end{tabular}
    }
  \caption{Model architecture for $WEANET^L$ for FSDKaggle-2019 dataset: All convolutional layers (except B6) are followed by batch norm and ReLU; B6 is followed by sigmoid activation.}%
  \label{tab:cnnarchfsd}
\end{table}%

Table \ref{tab:cnnarchfsd} shows the $WEANET$ architecture used for experiments on FSDKaggle-2019 dataset. $WEANET^L$ is just a lighter version of the one shown in Table 1 in the main paper. To keep things simple, we also use a simpler parameter-free mapping function $g()$. We use global average pooling as $g()$, which takes an average of segment level outputs to produce recording level output.

\section{Class-wise performance for Audioset}
\textbf{Figure \ref{fig:classaud}} shows class-wise performance for different sound classes and the improvement obtained from the sequential self-teaching approach.  The blue bar shows performance obtained from base-model (a.k.a default teacher $\mathcal{N}^0$). The green or red bar shows the change in performance from SUSTAIN model (corresponding to $\mathcal{N}^4$) in Table 4 from main text. The classes have been sorted by change in performance, with maximum improvement for first bar in top plot and maximum reduction in \emph{Vibraphone} class in right most bar of bottommost plot. 

We see that classes such as \emph{Zing, Moo, Cattle, Owl, Yodeling} (first 5 bars in topmost plot), get an absolute improvement of up to $0.16$ to $0.19$ in MAP, leading to $40-60\%$ improvement in relative sense. As mentioned in the main text, there are few classes such as \emph{Mouse, Squeal, Rattle} for which performance improves by more than 100\%. Overall, \emph{Bagpipes} sounds are easiest to recognize and we achieve an AP of 0.931 for it. \emph{Squish} on the other hand is hardest to recognize with an AP of 0.02.

\begin{figure*}[t]
      \centering
      \includegraphics[width=0.9\linewidth]{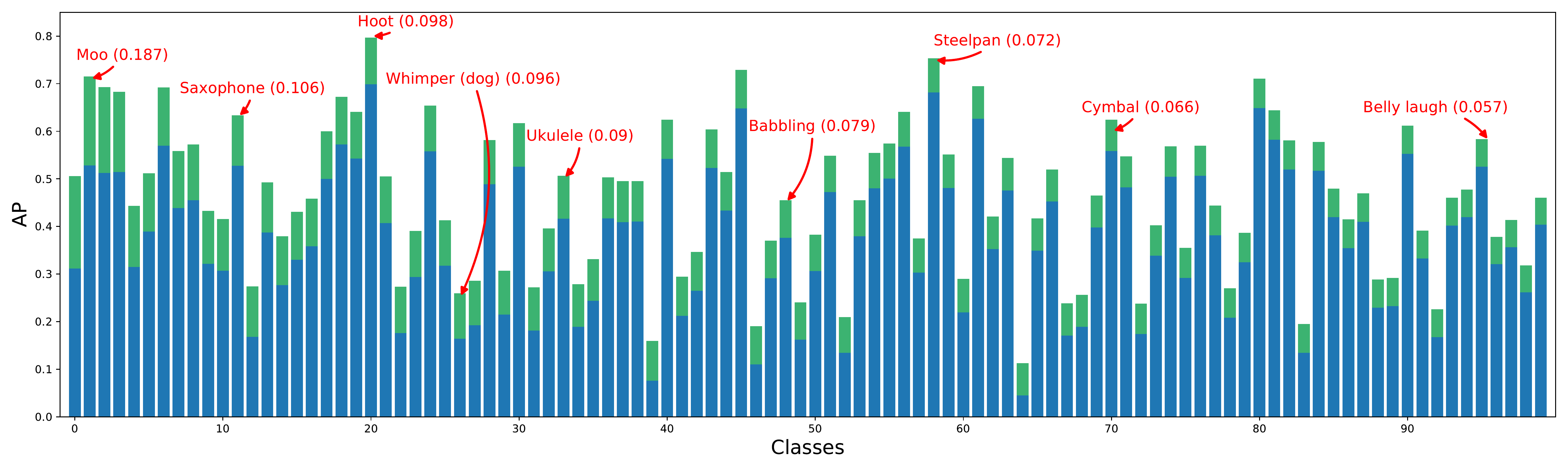}
      \includegraphics[width=0.9\linewidth]{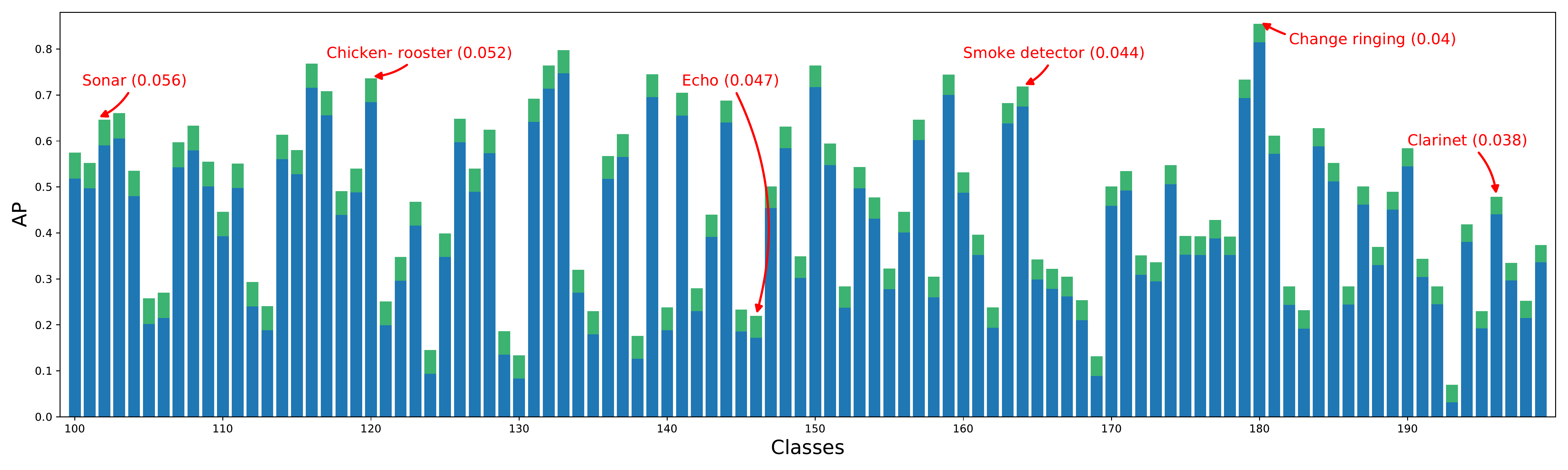}
      \includegraphics[width=0.9\linewidth]{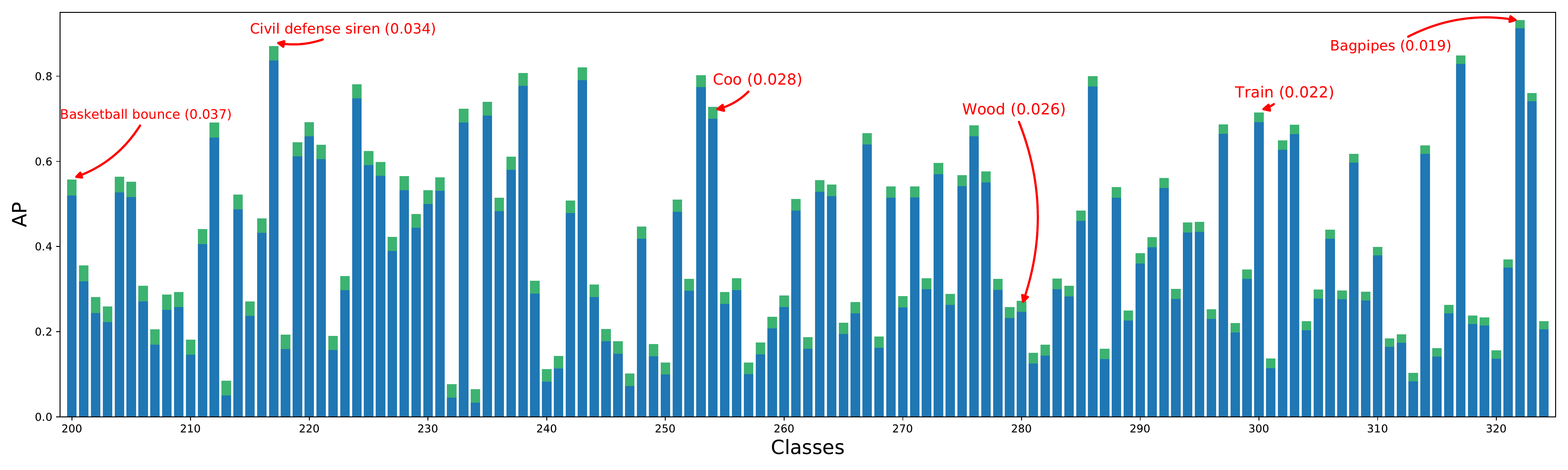}
      \includegraphics[width=0.9\linewidth]{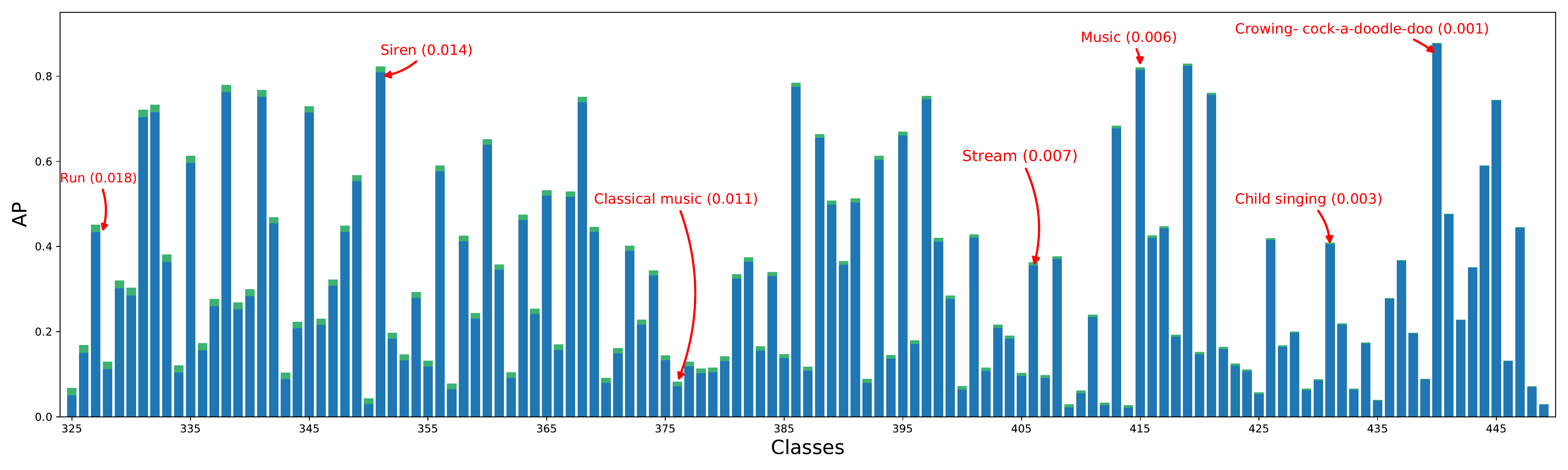}
      \includegraphics[width=0.9\linewidth]{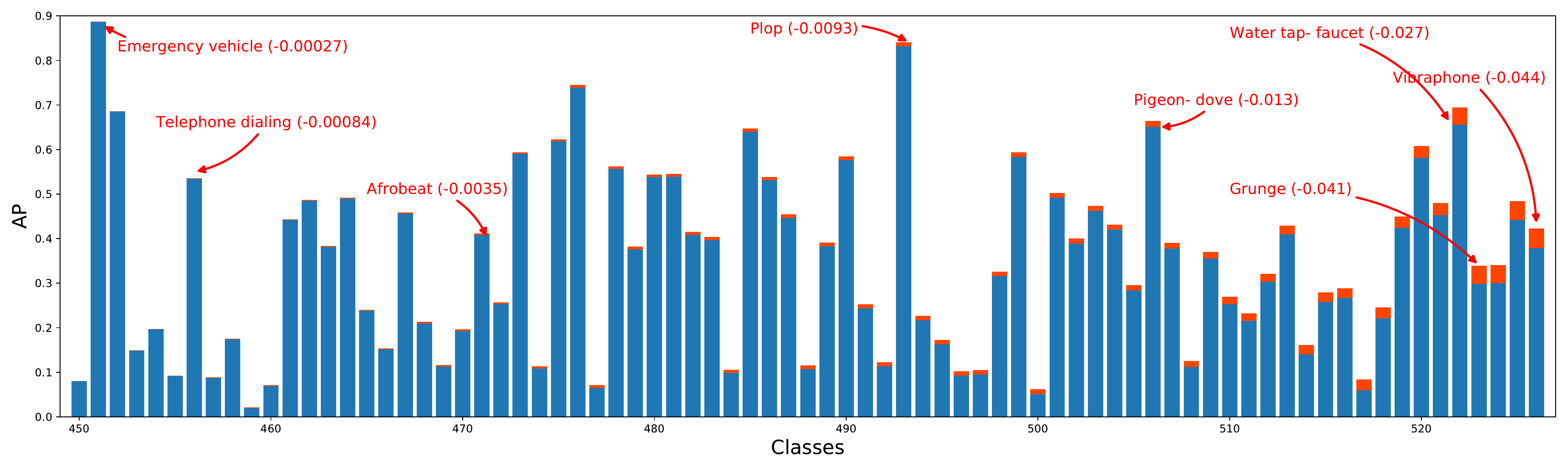}
    \caption{\textbf{Audioset} Class-wise AP and improvement in AP from SUSTAIN. The blue bar shows performance of $\mathcal{N}^0$, i.e. model trained only on available labels. The bar on top of each blue bar shows improvement (green) or deterioration (red) in performance from sequential teaching. Several classes (along with \textbf{absolute change} in performance ) have been annotated to bring out noteworthy observations. }%
    \label{fig:classaud}
\end{figure*}

\end{document}